\documentclass[a4paper,11pt]{article}

\usepackage{fullpage}

\usepackage{times}

\usepackage{amssymb,amsmath,stmaryrd}
\usepackage{color}
\usepackage[lined,boxed,commentsnumbered]{algorithm2e}
\newtheorem{defn}{Definition}[section]
\newtheorem{theorem}{Theorem}[section]
\newtheorem{lemma}[theorem]{Lemma}
\newtheorem{claim}[theorem]{Claim}

\newcommand{\sq}{\hbox{\rlap{$\sqcap$}$\sqcup$}}
\newcommand{\qed}{\hspace*{\fill}\sq}
\newenvironment{proof}{\noindent {\bf Proof.}\ }{\qed\par\vskip 4mm\par}

\newcommand{\bigO}{\mathcal{O}}

\title{Short sequences of improvement moves lead to approximate equilibria in constraint satisfaction games}

\author{Ioannis Caragiannis\thanks{Computer Technology Institute and Press ``Diophantus'' \& Department of Computer Engineering and Informatics, University of Patras, 26500 Rio, Greece. Email: {\tt caragian@ceid.upatras.gr}. The work is partially supported by the European Social Fund and Greek national funds through the research funding program Thales on ``Algorithmic Game Theory''.} \and Angelo Fanelli\thanks{CNRS, Caen, France. Email: {\tt angelo.fanelli@gmail.com}} \and Nick Gravin\thanks{Microsoft Research New England, Cambridge, MA, USA. Email: {\tt ngravin@microsoft.com}}}

\date{}

\begin{document}

\maketitle

\begin{abstract}
We present an algorithm that computes approximate pure Nash equilibria in a broad class of constraint satisfaction games that generalize the well-known cut and party affiliation games. Our results improve previous ones by Bhalgat et al.~(EC 10) in terms of the obtained approximation guarantee. More importantly, our algorithm identifies a polynomially-long sequence of improvement moves from any initial state to an approximate equilibrium in these games. The existence of such short sequences is an interesting structural property which, to the best of our knowledge, was not known before. Our techniques adapt and extend our previous work for congestion games (FOCS 11) but the current analysis is considerably simpler.

\medskip\noindent{\bf Keywords:} algorithmic game theory, complexity of equilibria, pure Nash equilibrium, potential games, constraint satisfaction
\end{abstract}

\section{Introduction}
Constraint satisfaction games are generalizations of the well-known cut games and party affiliation games. In a constraint satisfaction game, there is a set of boolean variables and a set of weighted constraints; each constraint depends on some of these variables. Each player controls the value of a distinct variable and has two possible strategies: setting the value of the variable to either $0$ (false) or $1$ (true). The payoff (or utility) of a player is the total weight in satisfied constraints where her variable appears. Constraint satisfaction games are potential games. The total weight of satisfied constraints serves as an exact potential function in the sense that the difference in the potential between two states that differ in the strategy of a single player equals the change in the utility of that player. Hence, pure Nash equilibria (i.e., states in which no player has an incentive to unilaterally move in order to improve her utility) can be computed by solving the local search problem (see \cite{MAK07} for a theoretical treatment of local search) of computing a local maximum of the potential function. Unfortunately, this is a computationally-hard problem \cite{SY91}. In this paper, we resort to the question of whether relaxed solution concepts --- namely, approximate (pure Nash) equilibria --- can be computed efficiently.

In particular, we consider constraint satisfaction games where each constraint depends on the value of at most $k$ variables and has the property that its value can change from false to true by a unilateral change in any of its variables. In general, we refer to such games as $P_k$--{\sc Flip} games following the terminology of Bhalgat et al.~\cite{BhalgatCK10}. Particular examples of this type of constraints include ``parity'' and ``not--all--equal'' constraints. An odd (respectively, even) parity constraint requires that the number of its true variables is odd (respectively, even). A not-all-equal constraint consists of literals (i.e., variables or their negations) and requires that at least two of its literals have different values. We refer to $P_k$--{\sc Flip} games consisting of parity constraints as {\sc Parity}--$k$--{\sc Flip} games; $P_k$--{\sc Flip} games with not--all--equal constraints with at least $\bar{k}$ literals are called {\sc nae}--$(\bar{k},k)$--{\sc Flip} games. Party affiliation games are {\sc Parity}--$2$--{\sc Flip} games and, in particular, cut games are {\sc Parity}--$2$--{\sc Flip} games with odd constraints or {\sc nae}--$(2,2)$--{\sc Flip} games whose constraints have no negative literals.

By adapting and extending our techniques in \cite{CFGS11} for congestion games, we present a polynomial-time algorithm that computes approximate equilibria in $P_k$--{\sc Flip} games. The approximation guarantee is related to the stretch $\theta$ of the potential function of games in a given class, defined as the maximum over all games in the class of the maximum ratio between the potential values in two equilibria. As we show, $P_k$--{\sc Flip} games have a stretch of $k+1$; hence, for general $P_k$--{\sc Flip} games, the approximation guarantee $\theta+\varepsilon$ of our algorithm improves a previous one of $2k-1+\varepsilon$ by Bhalgat et al.~\cite{BhalgatCK10} for $k\geq 3$. By bounding the stretch of {\sc nae}--$(\bar{k},k)$--{\sc Flip} and {\sc Parity}--$k$--{\sc Flip} games, we are able to show further improvements. For {\sc nae}--$(\bar{k},k)$--{\sc Flip} games, the approximation guarantee becomes $3+\varepsilon$ for $\bar{k}=2$ and $2+\varepsilon$ for $\bar{k}\geq 3$; these results improve a bound of $\frac{2\bar{k}}{\bar{k}-1}+\varepsilon$ from \cite{BhalgatCK10}. For {\sc Parity}--$k$--{\sc Flip} games with odd $k$, the approximation guarantee is $k+\varepsilon$. The running time of the algorithm is bounded by a polynomial of the number of players, $k$, and $1/\varepsilon$. Our analysis follows the same general structure of \cite{CFGS11} but uses different technical arguments and is considerably simpler due to the simplicity in the definition of $P_k$--{\sc Flip} games.

More importantly, for every initial state of the game, our algorithm identifies a polynomially-long sequence of improvement moves of the players that lead to an approximate equilibrium. The existence of such short sequence suggests an interesting structural property of $P_k$--{\sc Flip} games which, to the best of our knowledge, was not known before. Actually, Bhalgat et al.~\cite{BhalgatCK10} argue about the limitations of (uncoordinated) improvement move sequences by presenting a particular cut game in which any sequence of $\rho$-moves (i.e., moves that improve the utility of the moving player by a factor of at least $\rho$) from some states to any $\rho$--approximate equilibrium has exponential length for any $\rho\in [1,21/20)$. This negative result complements nicely with the structural property we prove.

Our algorithm is simple. Players are classified into blocks so that the players within the same block have polynomially-related maximum utility (i.e., total weight of the constraints a player can affect). Then, a set of phases is executed. In each phase the players in two consecutive blocks are allowed to move. The players in the block of higher maximum utility are allowed to make $p$-moves and the players of the other block are allowed to make $q$-moves. Then, the strategies of the players that were allowed to perform $p$-moves within a phase are irrevocably decided at its end. Clearly, this defines a sequence of improvement moves by the players. We show that by setting the parameters $q$ and $p$ appropriately, the algorithm terminates in polynomial time and, furthermore, the players whose strategies are irrevocably decided at the end of a phase will not be affected significantly by later moves. In order to do so, we select a value for parameter $p$ that is slightly higher than the stretch of the class of games to which the input game belongs and a value for parameter $q$ that is very close to $1$.

\paragraph{Related work.} Sch\"affer and Yannakakis \cite{SY91} proved that the problem of computing a pure Nash equilibrium in constraint satisfaction games is complete for the class PLS --- standing for polynomial local search --- that has been introduced by Johnson et al.~\cite{Johnson88}. The negative result of \cite{SY91} covers all games considered in the current work and have been strengthened in \cite{Klauck96,Krentel89} to capture instances in which each player participates in a constant number of constraints. Among the few rare non-trivial positive results is an algorithm by Poljac \cite{Poljac95} who shows that a local maximum of the potential function in cut games can be computed in polynomial time when each player participates in at most three constraints.

The algorithm of \cite{BhalgatCK10} for approximate equilibria in $P_k$--{\sc Flip} games has the following structure. Players are partitioned into layers in a similar way to the block partitioning that we use in the current paper. Then, a rearrangement phase moves players across blocks in order to guarantee that the total weight of constraints, in which a player $i$ participates together only with players in the same block or ones having lower maximum utility, is at least $1/k$ of player $i$'s maximum utility. This can be done in such a way that, eventually, each layer contains players with polynomially-related maximum utility. Then, a top-down layer dynamics phase takes place, where players within each layer play $(1+\varepsilon/k)$-moves in a restricted game among them until they reach an $(1+\varepsilon/k)$--approximate equilibrium in this restricted game. The authors of \cite{BhalgatCK10} show that the state computed in this way is a $(2k-1+\varepsilon)$--approximate equilibrium for the original game. They also present a variation of their algorithm for {\sc nae}--$(\bar{k},k)$--{\sc Flip} games that computes $(\frac{2\bar{k}}{\bar{k}-1}+\varepsilon)$--approximate equilibria. As the authors of \cite{BhalgatCK10} emphasize, in general, the moves during the top-down layer dynamics phase are not improvement moves in the original game. In contrast, our algorithm consists only of improvement moves.

Another class of potential games where the problem of computing an (approximate) equilibrium has received a lot of attention is that of congestion games. A classical potential function for these games has been defined by Rosenthal \cite{R73}. Fabrikant et al. \cite{FabrikantPT04} prove that computing a local minimum of this function (corresponding to a pure Nash equilibrium) is PLS-hard as well. Even worse, for sufficiently general congestion games, Skopalik and Voecking \cite{SkopalikV08} show that computing a $\rho$--approximate equilibrium is PLS-hard for every reasonable (i.e., polynomially-computable) value of $\rho$. In our previous work, we have presented an algorithm to compute $O(1)$-approximate equilibria for congestion games under mild assumptions for the structure of the game. The current paper adapts and extends the main algorithmic techniques in that paper, which have also been applied to (non-potential) weighted variants of congestion games in \cite{CFGS12}. Exact or almost exact equilibria can be computed in several special cases (e.g., see \cite{CS11,FabrikantPT04}).

We remark that, even though it is hard to compute exactly, a local optimum of a potential function can be approximated with extremely low precision under very mild assumptions \cite{OPS04}. This does not imply that equilibria can be approximated with a similar precision, as the negative results of \cite{SkopalikV08} show. Also, uncoordinated move sequences have been shown to reach states of high social value quickly \cite{AwerbuchAEMS08,BBM13,CMS06}, i.e., to states with low potential in the case of $P_k$--{\sc Flip} games. Unfortunately, these states are not approximate equilibria either, since some player typically has a high incentive to move.

\paragraph{Roadmap.} The rest of the paper is structured as follows. We begin with preliminary definitions in Section~\ref{sec:prelim}. Section \ref{sec:stretch} is devoted to our upper bounds on the stretch of $P_k$--{\sc Flip} games. The algorithm and the statement of our main result are presented in Section~\ref{sec:alg} and the analysis follows in Section~\ref{sec:proof}. We conclude with open problems in Section \ref{sec:open}.

\section{Preliminaries}\label{sec:prelim}
A constraint satisfaction game consists of a set $N$ of $n$ players, a set of at least $n$ boolean variables $V=\{s_1, s_2, ..., s_{|V|}\}$, and a set $C$ of constraints (henceforth called clauses) over the variables in $V$. Each clause $c\in C$ has a non-negative weight $w_c$. Player $j\in N$ controls the value of a distinct variable $s_j$ from $V$ and has two possible strategies: setting the value of $s_j$ to either $0$ (false), or $1$ (true). The variables of $V$ that are not controlled by any player (if any) are frozen to certain boolean values. A state $S$ of the game is simply a snapshot of variable values (or a snapshot of players strategies complemented with the fixed values of the frozen variables), i.e., $S=(s_1, s_2, ..., s_{|V|})$. Given a state $S$ of the game, we denote by $SAT(S)$ the set of satisfied clauses. For a subset of players $R\subseteq N$, we denote by $SAT_R(S)$ the subset of $SAT(S)$ that consists of clauses in which the variable of some player from $R$ appears. With some abuse of notation, we simplify $SAT_{\{j\}}(S)$ to $SAT_j(S)$. The utility of a player $j$ is the total weight of the true clauses in which her variable appears, i.e., $u_j(S) = \sum_{c\in SAT_j(S)}{w_c}$. We also denote by $C_R$ the set of clauses in which at least one player of $R$ participates and simplify $C_{\{j\}}$ to $C_j$. We use $U_j$ to denote the maximum possible utility that player $j$ might have, i.e., $U_j=\sum_{c\in C_j}{w_c}$.

Given a state $S=(s_1, s_2, ..., s_{|V|})$ and a player $j$, we denote by $(S_{-j},s'_j)$ the state obtained from $S$ when player $j$ unilaterally changes her strategy. This is an improvement move (or simply, a move) for player $j$ if her utility increases, i.e., $u_j(S_{-j},s'_j) > u_j(S)$. We call it a $\rho$--move when the utility increases by more than a factor of $\rho$, i.e., $u_j(S_{-j},s'_j) > \rho \cdot u_j(S)$. A state $S$ is a pure Nash equilibrium (or simply, an equilibrium) if no player has a move to make. Similarly, $S$ is a $\rho$--approximate (pure Nash) equilibrium if no player has a $\rho$-move.

We specifically consider clauses with the following property: any false clause can become true by changing the value of any of its variables. We will refer to games with clauses satisfying this property and with at most $k$ variables per clause as $P_k$--{\sc Flip} games. This class is broad enough and contains (generalizations of) several well-studied games such as cut games and party affiliation games. We are particularly interested in two subclasses of $P_k$--{\sc Flip} games. A {\sc nae}-clause contains literals (i.e., variables or their negations) and equals $1$ if and only if there are two literals with different values. We will refer to games consisting of {\sc nae}-clauses with at least $\bar{k}\geq 2$ and most $k$ literals as {\sc nae}--$(\bar{k},k)$--{\sc Flip} games. Observe that these games are $P_k$--{\sc Flip} games since changing the value of any variable that appears in a clause can change the value of the clause from $0$ to $1$. In {\sc Parity}--$k$--{\sc Flip} games, each clause is characterized as odd or even; an odd (respectively, even) clause is true if the number of its variables which are $1$ is odd (respectively, even). An important property of $P_k$--{\sc Flip} games is that for any state $S$ and any player $j$, it holds that $U_j\leq u_j(S)+u_j(S_{-j},s'_j)$.

Given a state $S$ of a $P_k$--{\sc Flip} game, we denote by $\Phi(S)$ the total weight of all true clauses, i.e., $\Phi(S)=\sum_{c\in SAT_N(S)}{w_c}$. The function $\Phi$ is a potential function for this game. In particular, it has the remarkable property that for every two states $S$ and $(S_{-j},s'_j)$ differing only in the strategy of player $j$, the difference of the potential is equal to the difference of the utility of player $j$, i.e., $\Phi(S)-\Phi(S_{-j},s'_j) = u_j(S)-u_j(S_{-j},s'_j)$.

In the following, we will be often considering sequences of moves in which only players in a certain subset $R\subseteq N$ are allowed to move. We can view such moves as moves in a subgame among the players in $R$, with the set of clauses $C_R$ (each clause in $C_R$ has the same weight as in the original game), and with fixed values for the variables that are not controlled by players in $R$. Observe that any subgame of a $P_k$--{\sc Flip} game is a $P_k$--{\sc Flip} game as well. Similarly, any subgame of a {\sc nae}--$(\bar{k},k)$--{\sc Flip} (respectively, {\sc Parity}--$k$--{\sc Flip}) game is a {\sc nae}--$(\bar{k},k)$--{\sc Flip} (respectively, {\sc Parity}--$k$--{\sc Flip}) game as well. The function $\Phi_R(S) = \sum_{c\in SAT_R(S)}{w_c}$ is an exact potential function for the subgame among the players in $R$. The next claim follows easily by the definitions.

\begin{claim}\label{cl:easy}
For every state $S$ of a $P_k$--{\sc Flip} game and any set of players $R\subseteq N$, it holds that $\Phi_R(S) \leq \sum_{j\in R}{u_j(S)} \leq k\Phi_R(S)$. Furthermore, for every set of players $R'\subseteq R$, it holds that $\Phi_{R'}(S) \leq \Phi_{R}(S)$. 
\end{claim}

\begin{proof}
The first inequality follows since every clause that contributes to the sum $\sum_{c\in SAT_R(S)}{w_c}$ (which is equal to $\Phi_R(S)$) contributes at least once and at most $k$ times to the sum $\sum_{j\in R}\sum_{c\in SAT_j(S)}{w_c}$ (which is equal to $\sum_{j\in R}{u_j(S)}$). The second one follows trivially since $SAT_{R'}(S)\subseteq SAT_R(S)$.
\end{proof}

\section{The stretch of $P_k$--{\sc Flip} games}\label{sec:stretch}
The approximation guarantee of our algorithm depends on a quantity related to the potential function of $P_k$--{\sc Flip} games that we call the stretch.

\begin{defn}
Given $\eta\geq 0$, the $(1+\eta)$-stretch of a $P_k$--{\sc Flip} game is the ratio between the maximum and the minimum value of the potential function taken over all $(1+\eta)$-approximate pure Nash equilibria of the game.
\end{defn}

We use the term stretch as a synonym of $1$-stretch; observe that it is simply the ratio between the maximum and minimum potentials of (exact) equilibria. In Theorem \ref{thm:stretch}, we present upper bounds on the $(1+\eta)$--stretch of $P_k$--{\sc Flip} games. Note that these bounds may be of independent interest; bounds on the stretch of congestion games from our previous work \cite{CFGS11} have been used by Piliouras et al.~\cite{PNS13} in order to quantify the price of anarchy of congestion games in settings with uncertainty where players have particular risk attitudes.
\begin{theorem}\label{thm:stretch}
For any $\eta>0$, the $(1+\eta)$--stretch of $P_k$--{\sc Flip} games, {\sc nae}--$(3,k)$--{\sc Flip} games, {\sc nae}--$(2,k)$--{\sc Flip} games, and {\sc Parity}--$k$--{\sc Flip} games with odd $k$ is at most $k+1+k\eta$, $2+k\eta$, $3+k\eta$, and $k+k\eta$, respectively.
\end{theorem}

\begin{proof}
Consider a $P_k$--{\sc Flip} (sub)game among players in a set $R$ and with a set of clauses $C_R$. Consider an $(1+\eta)$-approximate pure Nash equilibrium $S$ and let $S^*$ be a state that maximizes the potential function. Clearly, this state is an $(1+\eta)$--approximate equilibrium for every $\eta\geq 0$. Let $D\subseteq R$ be the set of players that use different strategies in $S$ and $S^*$. We denote by $C^i\subseteq C_R$ the set of clauses that contain exactly $i$ players from $D$ for $i=0, 1, ..., k$. We use $C^i_j$ to denote the subset of $C^i$ in which player $j$ participates. Let $SAT^i_R(S) = SAT_R(S) \cap C^i$. Also, denote by $\Lambda_j(S)$ the subset of $SAT_j(S)$ consisting of the clauses that would become false by changing the strategy of player $j\in D$ (to her strategy in $S^*$). Let $\lambda_c(S) = |\{j\in D: c\in \Lambda_j(S)\}|$ and $\lambda = \max_{c\in C_R}{\lambda_c(S)}$.

Since every player $j$ in $D$ has no $(1+\eta)$-move in state $S$, we have $(1+\eta)\cdot u_j(S) \geq u_j(S_{-j},s'_j)$ and, equivalently,
\begin{eqnarray*}
(1+\eta) \cdot \sum_{c\in SAT_j(S)}{w_c} &\geq & \sum_{c\in C_j\setminus \Lambda_j(S)}{w_c}.
\end{eqnarray*}
By adding $\sum_{c\in \Lambda(j)}{w_c}$ to both sides, we get
\begin{eqnarray*}
(1+\eta) \cdot \sum_{c\in SAT_j(S)}{w_c} +\sum_{c\in \Lambda_j(S)}{w_c} &\geq & \sum_{c\in C_j}{w_c}.
\end{eqnarray*}
By summing over all players in $D$, we obtain that
\begin{eqnarray}\label{eq:stretch}
\sum_{i\in[k]}{\sum_{c\in SAT^i_R(S)}}{(i+i\eta+\lambda_c(S)) w_c} & \geq & \sum_{i\in[k]}{\sum_{c\in C^i}{i w_c}},
\end{eqnarray}
where $[k]$ denotes the set of integers $\{1, 2, ..., k\}$. 

Now, the potential of state $S^*$ is not higher than the total weight of all clauses in $C_R\setminus C^0$ plus the weight in satisfied clauses of $C^0$ (these clauses are satisfied in both states $S$ and $S^*$). Hence,
\begin{eqnarray*}
\Phi(S^*)&\leq &\sum_{c\in C_R\setminus C^0}{w_c}+\sum_{c\in SAT^0_R(S)}{w_c}\\
& = & \sum_{i\in[k]}{\sum_{c\in C^i}{i w_c}} - \sum_{i\in[k]}{\sum_{c\in C^i}{(i-1) w_c}}+\sum_{c\in SAT^0_R(S)}{w_c}\\
& \leq & \sum_{i\in[k]}{\sum_{c\in SAT^i_R(S)}{(i+i\eta+\lambda_c(S)) w_c}} - \sum_{i\in[k]}{\sum_{c\in SAT^i_R(S)}{(i-1) w_c}}+\sum_{c\in SAT^0_R(S)}{w_c}\\
&= & \sum_{i\in[k]}{\sum_{c\in SAT^i_R(S)}{(\lambda_c(S)+1+i\eta) w_c}}+\sum_{c\in SAT^0_R(S)}{w_c}\\
&\leq & (\lambda+1+k \eta) \cdot \sum_{i=0}^{k}{\sum_{c\in SAT^i_R(S)}{w_c}}\\
&=&  (\lambda+1+k \eta) \cdot \sum_{c\in SAT_R(S)}{w_c}\\
&=& (\lambda+1+k \eta)\cdot \Phi(S).
\end{eqnarray*}
The second inequality in the above derivation follows from (\ref{eq:stretch}) and from the observation that $SAT^i_R\subseteq C^i$ for every $i\in [k]$. The last inequality follows by the definition of $\lambda$. Now, for general $P_k$--{\sc Flip} games, the theorem is obvious since $\lambda\leq k$.

In order to prove the next two statements, we need an additional simple observation. For any true clause $c$ that is a {\sc nae}-clause with at least three literals, $c$ can become false by an unilateral change in at most one variable (i.e., $\lambda_c(S)\leq 1$ and, consequently, $\lambda\leq 1$). A true {\sc nae}-clause with two literals can become false by a unilateral change in any of its two variables (hence, $\lambda \leq 2$). 

Finally, in order to prove the bound on the stretch of {\sc Parity}--$k$--{\sc Flip} games with odd $k$, we first observe that a clause of $SAT_R^k(S)$ is not satisfied in $S^*$, since changing $k$ (an odd number) variables changes the parity of the whole Parity clause. Hence, we can repeat the last derivation starting with the stronger bound $\Phi_R(S^*) \leq \sum_{c\in C_R\setminus SAT_R^k(S)}{w_c}$ and obtain the improved upper bound of $(\lambda+k\eta) \cdot \Phi_R(S)$ on $\Phi_R(S^*)$.
\end{proof}

The bounds in Theorem \ref{thm:stretch} are tight; we show this for $\eta=0$ with four examples. First, let $k\geq 2$ and consider the a unit-weight clause with the variables $x_1, x_2, ..., x_k$ that is true if and only if the number of variables with value $1$ is either zero or has the same parity with $k$ (it can be easily seen that this constraint satisfies the property required by $P_k$--{\sc Flip} games). There are $k$ additional even clauses, each containing only the variable $x_i$ for $i=1, ..., k$. The state in which all players play $1$ is an equilibrium with potential $1$ while the state in which all players play $0$ has potential $k+1$. Second, consider the {\sc nae}--$(3,3)$--{\sc Flip} game with two players controlling the values of the variables $x$ and $y$ and two unit-weight {\sc nae}-clauses $c_1=(0,x,y)$ and $c_2=(y,1,1)$. The state in which the players play $x=0$ and $y=1$ is an equilibrium with potential $1$ while the state with $x=1$ and $y=0$ has potential $2$. Third, consider the $2$-player {\sc nae}--$(2,k)$--{\sc Flip} game with the three unit-weight clauses $c_1=(0,x)$, $c_2=(x,y)$, and $c_3=(y,1)$. The state in which the players play $x=0$ and $y=1$ is an equilibrium with potential $1$ while the state with $x=1$ and $y=0$ has potential $3$. Finally, for odd $k$, consider the $(k-1)$-player {\sc Parity}--$k$--{\sc Flip} game with a unit-weight even clause $c=(x_1, x_2, ..., x_{k-1},0)$ and $k-1$ additional unit-weight odd clauses, each containing only the variable $x_i$ for $i=1, ..., k-1$. The state in which all players play $0$ is an equilibrium with potential $1$ while the state in which all players play $1$ has potential $k$.

In the following, we use the notation $\theta(1+\eta)$ to denote our upper bound on the $(1+\eta)$-stretch of $P_k$--{\sc Flip} games (and clarify when we refer to the stretch of particular subclasses of $P_k$--{\sc Flip} games). We use simply $\theta$ to denote the upper bound on the $1$-stretch.

\section{The algorithm}\label{sec:alg}
The pseudocode of our algorithm appears below as Algorithm 1. We supplement this formal description with a detailed line--by--line explanation.

\IncMargin{2em}
\RestyleAlgo{boxed}
\LinesNumbered
\begin{algorithm}
\SetKwInOut{Input}{Input}\SetKwInOut{Output}{Output}

\Input{A $P_k$--{\sc Flip} game ${\cal G}$ with a set $N$ of $n$ players, an arbitrary initial state $S_{\mbox{\small in}}$, and $\varepsilon\in (0,1]$}
\Output{A state $S_{\mbox{\small out}}$ of ${\cal G}$}
   $q\leftarrow 1+\frac{\varepsilon}{3k}$\;
   $p\leftarrow\theta(q)+\varepsilon/3$\;
   $\Delta=200 p^3nk/\varepsilon^2$\;
   Set $U_{\min}\leftarrow\min_{j\in N}{U_j}$, $U_{\max}\leftarrow\max_{j\in N}{U_j}$, and
       $m\leftarrow1+\left\lfloor \log_{\Delta}{\left(U_{\max}/U_{\min}\right)}\right\rfloor$\;
   (Implicitly) partition players into \emph{blocks} $B_1, B_2,\ldots, B_m$, such that $j\in B_i$ implies that
       $U_j\in\left(U_{\max} \Delta^{-i},U_{\max}\Delta^{1-i}\right]$\;
   $S\leftarrow S_{\mbox{\small in}}$\;
   \While{there exists a player $j\in B_1$ such that $u_j(S_{-j},s'_j)>q\cdot u_j(S)$}
         {
         $S\leftarrow (S_{-j},s'_j)$\;
         }
   \For{phase $i\leftarrow 1$ \KwTo $m-1$ such that $B_i\not=\emptyset$}{
         \While{there exists a player $j$ that either belongs to $B_i$ and satisfies $u_j(S_{-j},s'_j)>p\cdot u_j(S)$ or belongs to $B_{i+1}$ and satisfies $u_j(S_{-j},s'_j)>q\cdot u_j(S)$}
         {
         $S\leftarrow (S_{-j},s'_j)$\;
   }}
   $S_{\mbox{\small out}} \leftarrow S$\;

\caption{Computing approximate equilibria in $P_k$--{\sc Flip} games.}\label{alg}
\end{algorithm}
\DecMargin{2em}
The algorithm takes as input a $P_k$--{\sc Flip} game ${\cal G}$ with $n$ players, an initial state $S_{\mbox{\small in}}$, and an accuracy parameter $\varepsilon\in (0,1]$. Starting from state $S_{\mbox{\small in}}$, it identifies a sequence of moves that lead to a state $S_{\mbox{\small out}}$; this is the output of the algorithm. As we will prove later, $S_{\mbox{\small out}}$ is an approximate equilibrium. The algorithm starts (lines 1 and 2) by setting the values of parameters $q$ and $p$. Parameter $q$ has a value very close to $1$ (namely, $q=1+\frac{\varepsilon}{3k}$) and parameter $p$ has a value slightly higher than the $q$-stretch of the class to which the input game belongs (namely, $p=\theta(q)+\varepsilon/3$). In particular, using our upper bounds on $\theta(q)$ from Theorem \ref{thm:stretch}, $p$ is set to be $k+1+2\varepsilon/3$ in general, $2+2\varepsilon/3$ if ${\cal G}$ is a {\sc nae}--$(3,k)$--{\sc Flip} game, $3+2\varepsilon/3$ if it is a {\sc nae}--$(2,k)$--{\sc Flip} games, and $k+2\varepsilon/3$ if it is a {\sc Parity}--$k$ game and $k$ is odd. The algorithm also sets the value of parameter $\Delta$ to be a polynomial depending on $n$, $k$, $p$, and $1/\varepsilon$ (line 3). Then (lines 4-5), it implicitly partitions the players into blocks $B_1$, $B_2$, ..., $B_m$ according to their maximum utility. Denoting by $U_{\max}$ the maximum values among all players' maximum utilities, block $B_i$ consists of the players $j$ with maximum utility $U_j\in (U_{\max}\Delta^{-i},U_{\max}\Delta^{1-i}]$. By the definition of $\Delta$, the players in the same block have polynomially related maximum utilities.

The sequence of moves from state $S_{\mbox{\small in}}$ to state $S_{\mbox{\small out}}$ is computed by the code in the lines 6-15. The subsequence of moves described in lines 7-9 constitutes phase $0$. During phase $0$, the players in block $B_1$ make $q$-moves. After that, each phase $i$ for $i\geq 1$ consists of $p$-moves of players in block $B_i$ and $q$-moves of players in block $B_{i+1}$. Strategies of players in block $B_i$ are irrevocably decided at the end of phase $i$.

We are ready to state our main result which we will prove in the next section.

\begin{theorem}\label{thm:main}
On input a $P_k$--{\sc Flip} game ${\cal G}$ with $n$ players, an initial state $S_{\mbox{\small in}}$, and $\varepsilon\in (0,1]$, Algorithm 1 computes a sequence of at most $\mbox{poly}(n,k,1/e)$ moves that starts from $S_{\mbox{\small in}}$ and converges to a $(k+1+\varepsilon)$--approximate pure Nash equilibrium $S_{\mbox{\small out}}$. The approximation guarantee is at most $2+\varepsilon$ when ${\cal G}$ is a {\sc nae}--$(3,k)$--{\sc Flip} game, at most $3+\varepsilon$ when it is a {\sc nae}--$(2,k)$--{\sc Flip} games, and at most $k+\varepsilon$ when it is a {\sc Parity}--$k$--{\sc Flip} game and $k$ is odd.
\end{theorem}

\section{Proof of Theorem \ref{thm:main}}\label{sec:proof}
Before presenting the proof of Theorem \ref{thm:main}, we give some intuition behind our analysis. The analysis uses two properties that are formally stated in Lemma \ref{lem:key}. What this lemma essentially says is that, during each phase, the total utility of the moving players as well as an increase in the potential of the subgame among these players are small. The first property is used in Lemma \ref{lem:negligible-effect} to prove that, once the strategy of a player is irrevocably decided, later phases may have only a negligible effect on her. And since no player has a $p$-move to make at the end of the phase when her strategy is decided, she cannot improve her utility by a factor of (almost) $p$ until the end of the algorithm. Together with the fact that each player's move increases her utility by some non-negligible amount, the second property is used in Lemma \ref{lem:rtime} to bound the total number of moves.

In our analysis, we denote by $S^i$ the state reached at the end of phase $i\geq 0$, i.e., $S_{\mbox{\small out}}=S^{m-1}$. We also denote by $R_i$ the set of players that move during phase $i$. We also denote the upper boundary of block $B_i$ by $W_i$ and by $W_{m+1}$ the lower boundary of block $B_m$, i.e., $W_i=U_{\max}\Delta^{1-i}$ for $i=1, 2, ..., m+1$. So, the players of block $B_i$ are those with maximum utility $U_j\in (W_{i+1},W_i]$.
\begin{lemma}\label{lem:key}
For every phase $i\geq 1$, it holds that
\begin{enumerate}
\item $\sum_{j\in R_i}{U_j} \leq 10pknW_{i+1}/\varepsilon$
\item $\Phi_{R_i}(S^i)-\Phi_{R_i}(S^{i-1}) \leq 3p^2 n W_{i+1}/\varepsilon$.
\end{enumerate}
\end{lemma}

\begin{proof}
First observe that players not in $R_i$ have the same set of strategies in states $S^{i-1}$ and $S^i$. Furthermore, the total weight of clauses depending on variables that are controlled by players from $R_i\cap B_{i+1}$ is at most $nW_{i+1}$. Hence, by the definition of the subgame potential, we have that the potential of the state $(S^{i-1}_{-R_{i}\cap B_i},S^i_{R_i\cap B_i})$ in which the players in $R_i\cap B_i$ play their strategies in state $S^i$ and the remaining players play their strategies in $S^{i-1}$ satisfies
\begin{eqnarray}\label{eq:1}
\Phi_{R_i\cap B_i}(S^{i-1}_{-R_i\cap B_i},S^i_{R_i\cap B_i}) &\geq & \Phi_{R_i}(S^i)-nW_{i+1}.
\end{eqnarray}
We will use inequality (\ref{eq:1}) in the proof of the next claim that provides a bound on the potential $\Phi_{R_i}(S^{i-1})$ as well as later in the current proof.

\begin{claim}\label{cl:vanila}
$\Phi_{R_i}(S^{i-1}) \leq 3p n W_{i+1}/\varepsilon$.
\end{claim}
\begin{proof}
We assume on the contrary that $\Phi_{R_i}(S^{i-1}) > 3p n W_{i+1}/\varepsilon$ and we are going to conclude that the potential of the state $(S^{i-1}_{-R_{i}\cap B_i},S^i_{R_i\cap B_i})$ satisfies $\Phi_{R_i\cap B_i}(S^{i-1}_{-R_{i}\cap B_i},S^i_{R_i\cap B_i})>\theta(q) \cdot \Phi_{R_i\cap B_i}(S^{i-1})$. By Theorem \ref{thm:stretch}, this would contradict the fact that $S^{i-1}$ is the output of phase $i-1$, i.e., a $q$-approximate equilibrium of the subgame among the players in $R_i\cap B_i$, since there is another $q$--approximate equilibrium (the one that can be reached from $(S^{i-1}_{-R_{i}\cap B_i},S^i_{R_i\cap B_i})$ with $q$-moves by the players in $R_i\cap B_i$) with a potential that is higher than $\theta(q)$ times the potential at state $S^{i-1}$.

We denote by $\ell(j)$ the utility of  player $j\in R_i\cap B_i$ right after she makes her last move in phase $i$. Then we have
\begin{eqnarray}\label{eq:2}
\Phi_{R_i}(S^i) - \Phi_{R_i}(S^{i-1}) &\geq & (1-1/p) \cdot \sum_{j\in R_i\cap B_i}{\ell(j)}.
\end{eqnarray}
Indeed, the last move of a player $j\in R_i\cap B_i$ increases her utility by a factor of at least $p$ and the difference $\Phi_{R_i}(S^i) - \Phi_{R_i}(S^{i-1})$ equals to the total increase in the utility of the deviating players within the phase.

Furthermore, we claim that
\begin{eqnarray}\label{eq:3}
\sum_{j\in R_i\cap B_i}{\ell(j)}+nW_{i+1} &\geq & \Phi_{R_i}(S^i).
\end{eqnarray}
To see why \eqref{eq:3} is true, observe that the right-hand side is the sum of the weights of the clauses in $SAT_{R_i}(S^i)$. The term $nW_{i+1}$ is an upper bound on the total weight of the clauses in $SAT_{R_i\cap B_{i+1}}(S^i)$. The weight of each of the remaining ones (i.e., the clauses in $SAT_{R_i}(S^i)\setminus SAT_{R_i\cap B_{i+1}}(S^i)$) is accounted for at least once in the sum $\sum_{j\in R_i\cap B_i}{\ell(j)}$, as part of the utility of some player from $R_i\cap B_i$ after her last move.

By (\ref{eq:2}) and (\ref{eq:3}) (i.e., by multiplying (\ref{eq:2}) by $p$ and (\ref{eq:3}) by $p-1$ and summing them), we obtain that
\begin{eqnarray}\label{eq:4}
\Phi_{R_i}(S^i) & \geq & p\cdot \Phi_{R_i}(S^{i-1})-(p-1)nW_{i+1}.
\end{eqnarray}
Hence, using (\ref{eq:1}), (\ref{eq:4}), the definition of $p$, and the second inequality of Claim \ref{cl:easy}, we obtain
\begin{eqnarray*}
\Phi_{R_i\cap B_i}(S^{i-1}_{-R_i\cap B_i},S^i_{R_i\cap B_i}) & \geq & \Phi_{R_i}(S^i)-nW_{i+1}\\
& \geq & p \cdot \Phi_{R_i}(S^{i-1}) - pnW_{i+1}\\
&>& (p-\varepsilon/3) \cdot \Phi_{R_i}(S^{i-1})\\
&\geq & \theta(q) \cdot \Phi_{R_i\cap B_i}(S^{i-1}).
\end{eqnarray*}
We have obtained the desired contradiction. 
\end{proof}

Using the observation that no player in $R_i\cap B_i$ has a $q$-move to make at the end of phase $i-1$ (i.e., at state $S^{i-1}$) as well as the first inequality of Claim \ref{cl:easy}, we obtain that
\begin{eqnarray*}
\sum_{j\in R_i\cap B_i}{U_j} &\leq & \sum_{j\in R_i\cap B_i}{\left(u_j (S^{i-1}) + u_j (S^{i-1}_{-j},s'_j)\right)}\\
&\leq & \sum_{j\in R_i\cap B_i}{(1+q)u_j (S^{i-1})}\\
&\leq & (1+q)k\cdot \Phi_{R_i\cap B_i}(S^{i-1})\\
&\leq & 9pknW_{i+1}/\varepsilon.
\end{eqnarray*}
The proof of the first inequality in the statement of the lemma follows by observing that the total utility of the players in $R_i\cap B_{i+1}$ is at most $nW_{i+1}$.

In order to prove the second inequality we use inequality \eqref{eq:1}, the $q$-stretch bound for the subgame among the players in $R_i\cap B_i$, the fact that $\theta(q)\leq p$, the second inequality of Claim \ref{cl:easy}, and the bound on $\Phi_{R_i}(S^{i-1})$ from Claim~\ref{cl:vanila}.
\begin{eqnarray*}
\Phi_{R_i}(S^i)-\Phi_{R_i}(S^{i-1}) &\leq & \Phi_{R_i\cap B_i}(S^{i-1}_{-R_i\cap B_i},S^i_{R_i\cap B_i})-\Phi_{R_i}(S^{i-1}) +nW_{i+1}\\
&\leq & \theta(q) \cdot \Phi_{R_i\cap B_i}(S^{i-1})-\Phi_{R_i}(S^{i-1}) +nW_{i+1}\\
&\leq & (p-1)\cdot \Phi_{R_i}(S^{i-1})+nW_{i+1}\\
&\leq & 3p^2n W_{i+1}/\varepsilon.
\end{eqnarray*}

\end{proof}

The first property of Lemma \ref{lem:key} indicates that the total weight of the moving players in phase $i$ is significantly smaller than the upper boundary of block $B_i$. In Lemma \ref{lem:negligible-effect} we combine this with the fact that the upper boundary of subsequent blocks decreases exponentially and formally prove that, after the strategy of a player is irrevocably decided, subsequent phases may have only a negligible effect on her. Recall that $\theta$ is the stretch of the class of games to which the input game belongs to and equals $k+1$ for $P_k$--{\sc Flip} games, $3$ for {\sc nae}--$(2,k)$--{\sc Flip} games, and $2$ for {\sc nae}--$(3,k)$--{\sc Flip} games, and $k-1$ for {\sc Parity}--$k$--{\sc Flip} games with odd $k$.

\begin{lemma}\label{lem:negligible-effect}
The state $S_{\mbox{\small out}}$ is a $(\theta+\varepsilon)$--approximate pure Nash equilibrium.
\end{lemma}

\begin{proof}
By the definition of phase $m-1$, the players in blocks $B_{m-1}$ and $B_m$ have no $p$-move to make at the end of phase $m-1$. We will consider a player $j$ belonging to block $B_t$ whose strategy is irrevocably decided at the end of phase $t$ with $t\leq m-2$, and will show that she has no $(p+\varepsilon/3)$-move to make at the end of phase $m-1$ (i.e., at state $S^{m-1}=S_{\mbox{\small out}}$). The lemma will then follow since $p+\varepsilon/3=\theta(1+\frac{\varepsilon}{3k})+2\varepsilon/3=\theta+\varepsilon$.

Let $s_j$ be the strategy used by player $j$ at the end of phase $t$. Using Lemma \ref{lem:key} and the definition of the block boundaries, we can bound the quantity $\sum_{i=t+1}^{m-1}{\sum_{r\in R_i}{U_{r}}}$. Thus, we get an upper bound on the total weight of clauses with players that move in phases $t+1, ..., m-1$, as follows:
\begin{eqnarray}\nonumber
\sum_{i=t+1}^{m-1}{\sum_{r\in R_i}{U_{r}}} &\leq & \sum_{i=t+1}^{m-1}{10pnkW_{i+1}/\varepsilon}\\\nonumber
&\leq & \frac{10pnkW_{t+1}}{\varepsilon}\sum_{i=1}^{\infty}{\Delta^{-i}}\\\nonumber
&=& \frac{10pnkW_{t+1}}{\varepsilon (\Delta-1)}\\\label{eq:all-weights}
&\leq & \frac{W_{t+1}\varepsilon}{10p^2}.
\end{eqnarray}
The last inequality follows by the definition of $\Delta$ and the fact that $\Delta-1\geq \Delta/2$.

Now observe that since player $j$ has no $p$-move at the end of phase $t$ (i.e., at state $S^t$), it holds that
$u_j(S^t)\geq u_j(S^t_{-j},s'_j)/p$ and $W_{t+1} \leq u_j(S^t)+u_j(S^t_{-j},s'_j) \leq (1+p)u_j(S^t)$, i.e., $u_j(S^t)\geq \frac{W_{t+1}}{1+p}$. Furthermore, during phases $t+1, ..., m-1$, the total change in the utility of player $j$ or in the utility player $j$ would have by deviating is at most $\sum_{i=t+1}^{m-1}{\sum_{r\in R_i}{U_{r}}}$. Using these observations and inequality (\ref{eq:all-weights}), we have
\begin{eqnarray*}
u_j(S^{m-1}) &\geq & u_j(S^t) - \sum_{i=t+1}^{m-1}{\sum_{r\in R_i}{U_{r}}}\\
&\geq & \frac{p}{p+\varepsilon/3}u_j(S^t)+\frac{\varepsilon/3}{p+\varepsilon/3}\frac{W_{t+1}}{1+p}-\sum_{i=t+1}^{m-1}{\sum_{r\in R_i}{U_{r}}}\\
&\geq & \frac{1}{p+\varepsilon/3}u_j(S^t_{-j},s'_j) +\frac{W_{t+1}\varepsilon}{5p(p+\varepsilon/3)}-\sum_{i=t+1}^{m-1}{\sum_{r\in R_i}{U_{r}}}\\
&\geq & \frac{1}{p+\varepsilon/3}u_j(S^{m-1}_{-j},s'_j) +\frac{W_{t+1}\varepsilon}{5p(p+\varepsilon/3)}-\left(1+\frac{1}{p+\varepsilon/3}\right)\sum_{i=t+1}^{m-1}{\sum_{r\in R_i}{U_{r}}}\\
&\geq & \frac{1}{p+\varepsilon/3}u_j(S^{m-1}_{-j},s'_j) +\frac{W_{t+1}\varepsilon}{5p(p+\varepsilon/3)}-\frac{2p}{p+\varepsilon/3}\sum_{i=t+1}^{m-1}{\sum_{r\in R_i}{U_{r}}}\\
&\geq & \frac{1}{p+\varepsilon/3}u_j(S^{m-1}_{-j},s'_j),
\end{eqnarray*}
as desired. In the third and fifth inequalities we have used the inequalities $3(1+p)\leq 5p$ and $p+1+\varepsilon/3 \leq 2p$ which follow since $p\geq 2$ and $\varepsilon\in (0,1]$. This completes the proof of the lemma.
\end{proof}

We conclude the proof of Theorem \ref{thm:main} by bounding the running time of the algorithm.
\begin{lemma}\label{lem:rtime}
On input of $P_k$--{\sc Flip} (in particular, {\sc nae}--$(\bar{k},k)$--{\sc Flip}) game, the algorithm identifies a sequence of at most $\bigO(n^3 k^7/\varepsilon^4)$ (in particular, $\bigO(n^3k^2/\varepsilon^4)$) moves.
\end{lemma}

\begin{proof}
Consider a moving player $j$ that belongs to block $i$ and let $u$ be her utility after she makes a move. Since this is a move in a $P_k$--{\sc Flip} game, $u\geq U_j/2$. Also, since it is at least an $\left(1+\frac{\varepsilon}{3k}\right)$-move (and since $k\geq 2$ and $\varepsilon \in (0,1]$), the potential improves by at least $u-\frac{u}{1+\frac{\varepsilon}{3k}} \geq \frac{\varepsilon U_j}{7k} \geq \frac{\varepsilon W_{i+1}}{7k}$.

We will bound the total number of moves by bounding the number of moves in each phase separately. Clearly, the increase in the potential during phase $0$ is $\Phi_{R_0}(S^0)-\Phi_{R_0}(S_{\mbox{\small in}}) \leq nW_{1}$. Hence, since only players in block $B_1$ move during phase $0$, it will end after at most $nW_{1}/\left(\frac{\varepsilon W_{2}}{7k}\right) = 7n k \Delta/\varepsilon = 1400p^3k^2n^2/\varepsilon^3$ moves. For phase $i\geq 1$, by Lemma \ref{lem:key}, we have $\Phi_{R_i}(S^{i-1})-\Phi_{R_i}(S^{i-1})\leq 3p^2nW_{i+1}/\varepsilon$. Since the moving players during this phase belong to blocks $B_i$ and $B_{i+1}$, the increase in the potential during each move is at least $\frac{\varepsilon W_{i+2}}{7k}$. Hence, the total number of moves during the phase is at most $\left(3p^2nW_{i+1}/\varepsilon\right)/\left(\frac{\varepsilon W_{i+2}}{7k}\right) = 21p^2nk\Delta/\varepsilon^2 = 4200 p^5 k^2 n^2/\varepsilon^4$.

In total, since the number of the phases that are executed by the algorithm after phase $0$ is at most $n$, the number of moves is at most $\bigO(n^3 p^5 k n^3/\varepsilon^4)$ and the lemma follows since $p\in \bigO(k)$ in general and $p=\bigO(1)$ in particular for {\sc nae}--$(\bar{k},k)$--{\sc Flip} games.
\end{proof}

\section{Open problems}\label{sec:open}
A challenging open problem is to improve the approximation guarantee of our algorithm. Our analysis indicates that a state with lower stretch at the beginning of each phase would allow us to use an even smaller value for parameter $p$ and, subsequently, to obtain a better approximation guarantee. One idea that comes immediately to mind is to replace the $q$-moves of the players of block $B_{i+1}$ within phase $i$ with the execution of an algorithm that computes states with approximately--optimal potential. For example, a random assignment to players of $B_{i+1}$ would yield a $2$-approximation to the potential of the subgame among them. Furthermore, for more structured $P_k$--{\sc Flip} games such as cut games, one might think to use the famous algorithm of \cite{GW95} that is based on semi-definite programming. Unfortunately, we do not see how to include these ideas into our algorithm at this point. The main difficulty is that the low-stretch property should hold for the subgame among the players that will move during the next phase which we do not know in advance. An algorithm that approximates the potential of all subgames simultaneously would be ideal here but, besides the local search approach implied by the $q$-moves, neither the random assignment nor the SDP-based algorithms satisfy this property.

Even if we could bypass these obstacle, our technique has limitations since computing states with low-stretch in $P_k$--{\sc Flip} games includes famous hard-to-approximate problems (e.g., see \cite{H01}). So, in order to compute almost exact equilibria, we need new techniques. Of course, we have no idea whether this is at all possible. To put the question differently, is there some inapproximability threshold for approximate equilibria? We remark that such negative statements are not known in the literature: the only known negative results are either specific to exact equilibria (such as the PLS-hardness results of \cite{FabrikantPT04,SY91}) or rule out any reasonable approximation guarantee in games with very general structure (e.g., in \cite{SkopalikV08}). We believe that such questions that are related to the computational complexity of approximate pure Nash equilibria deserve further attention.


\begin{thebibliography}{99}
\bibitem{AwerbuchAEMS08}
B. Awerbuch, Y. Azar, A. Epstein, V.~S. Mirrokni, and A. Skopalik. Fast convergence to nearly optimal solutions in potential games. In {\em Proceedings of the 9th ACM Conference on Electronic Commerce (EC)}, pages 264--273, 2008.

\bibitem{BBM13}
M.-F. Balcan, A. Blum, and Y. Mansour. Circumventing the price of anarchy: leading dynamics to good behavior. {\em SIAM Journal on Computing}, 42(1), pages 230--264, 2013.

\bibitem{BhalgatCK10}
A. Bhalgat, T. Chakraborty, and S. Khanna. Approximating pure Nash equilibrium in cut, party affiliation, and satisfiability games. In {\em Proceedings of the 11th ACM Conference on Electronic Commerce (EC)}, pages 73--82, 2010.

\bibitem{CFGS11}
I. Caragiannis, A. Fanelli, N. Gravin, and A. Skopalik. Efficient computation of approximate pure Nash equilibria in congestion games. In {\em Proceedings of the 52nd Annual IEEE Symposium on Foundations of Computer Science (FOCS)}, pages 532--541, 2011.

\bibitem{CFGS12}
I. Caragiannis, A. Fanelli, N. Gravin, and A. Skopalik. Approximate pure Nash equilibria in weighted congestion games: existence, efficient computation, and structure. In {\em Proceedings of the 13th ACM Conference on Electronic Commerce (EC)}, pages 284--301, 2012.


\bibitem{CS11}
S. Chien and A. Sinclair. Convergence to approximate Nash equilibria in congestion games. {\em Games and Economic Behavior}, 71(2), pages 315--327, 2011.

\bibitem{CMS06}
G. Christodoulou, V. S. Mirrokni, and A. Sidiropoulos. Convergence and approximation in potential games. {\em Theoretical Computer Science}, 438, pages 13--27, 2012.

\bibitem{FabrikantPT04}
A. Fabrikant, C.~H. Papadimitriou, and K. Talwar. The complexity of pure nash equilibria. In {\em Proceedings of the 36th Annual ACM Symposium on Theory of Computing (STOC)}, pages 604--612, 2004.

\bibitem{GW95}
M. X. Goemans and D. P. Williamson. Improved approximation algorithms for maximum sut and satisfiability problems using semidefinite programming. {\em Journal of the ACM}, 42(6), pages 1115--1145, 1995.

\bibitem{H01}
J. Hastad. Some optimal inapproximability results. {\em Journal of the ACM}, 48(4), pages 798--859, 2001.

\bibitem{Johnson88}
D.~S. Johnson, C.~H. Papadimitriou, and M. Yannakakis. How easy is local search? {\em Journal of Computer and System Sciences}, 37, pages 79--100, 1988.

\bibitem{Klauck96}
H. Klauck. On the hardness of global and local approximation. In {\em Proceedings of the 5th Scandinavian Workshop on Algorithm Theory (SWAT)}, LNCS 1097, Springer, pages 88--99, 1996.

\bibitem{Krentel89}
M. W. Krentel. Structure in locally optimal solutions. In {\em Proceedings of the 30th Annual Symposium on Foundations of Computer Science (FOCS)}, pages 216--221, 1989.

\bibitem{MAK07}
W. Michiels, E. Aarts, and J. Korst. Theoretical aspects of local search. {\em EATCS Monographs in Theoretical Computer Science}, Springer, 2007.


\bibitem{OPS04}
J. B. Orlin, A. P. Punnen, and A. S. Schulz. Approximate local search in combinatorial optimization. {\em SIAM Journal on Computing}, 33, pages 1201--1214, 2004.

\bibitem{PNS13}
G. Piliouras, E. Nikolova, and J. S. Shamma. Risk sensitivity of price of anarchy under uncertainty. In {\em Proceedings of the 13th ACM Conference on Electronic Commerce (EC)}, pages 715--732, 2013.

\bibitem{Poljac95}
S. Poljak. Integer linear programs and local search for max-cut. {\em SIAM Journal on Computing}, 24, pages 822--839, 1995.

\bibitem{R73}
R.~W. Rosenthal. A class of games possessing pure-strategy Nash equilibria. {\em International Journal of Game Theory}, 2, pages 65--67, 1973.

\bibitem{SY91}
A. A. Sch\"affer and M. Yannakakis. Simple local search problems that are hard to solve. {\em SIAM Journal on Computing}, 20, pages 56--87, 1991.

\bibitem{SkopalikV08}
A. Skopalik and B. V{\"o}cking. Inapproximability of pure Nash equilibria. In {\em Proceedings of the 41st Annual ACM Symposium on Theory of Computing (STOC)}, pages 355--364, 2008.

\end{thebibliography}
\end{document}